\documentclass[smallextended]{svjour3}       % onecolumn (second format)

\usepackage{amsmath,amsfonts,amssymb,mathrsfs}
\usepackage{graphicx}
\usepackage{latexsym}
\usepackage{xspace}
\usepackage{paralist}
\usepackage{algorithm}
\usepackage{algorithmic}
\usepackage{colortbl}
\usepackage{subfigure}
\usepackage{epstopdf}

\newcommand{\eat}[1]{}
\def \network{cognitive radio network\xspace}
\def \prob{\textsc{SpecCon}}
\def \fullprob{\textsc{Spectrum Connectivity}\xspace}
\def \unisat{\textsc{Uniform-SAT}\xspace}
\def \budget{antenna budget\xspace}
\def \specm{spectrum map\xspace}
\def \sm{\textsc{SpecMap}}
\def \specassign{spectrum assignment\xspace}
\def \realgraph{realization graph\xspace}
\def \pgraph{potential graph\xspace}
\def \pgraphs{potential graphs\xspace}
\def \pedge{potential edge\xspace}

\def \tw{\mathsf{tw}}

\renewcommand{\P}{{\rm P}\xspace}
\newcommand{\NP}{{\rm NP}\xspace}

\begin{document}
\title{On the Complexity of Connectivity in Cognitive Radio Networks Through Spectrum Assignment \thanks{A preliminary version of this paper appears in ALGOSENSORS 2012 \cite{algosensor12}. The main extensions are that we give more details of our results and investigate the special case where the potential graphs are of bounded treewidth.
}}

\titlerunning{On the Complexity of Connectivity in CRNs Through Spectrum Assignment}        % if too long for running head

\author{Hongyu Liang \and Tiancheng Lou \and Haisheng Tan
\and  Yuexuan Wang \and Dongxiao Yu }
\institute{Hongyu Liang, Haisheng Tan(Corresponding Author), Yuexuan Wang \at
Institute for Theoretical Computer Science, IIIS, Tsinghua University, Beijing, China\\
\email {lianghy08@mails.tsinghua.edu.cn, \{tan,wangyuexuan\}@mail.tsinghua.edu.cn}\\
\and
Tiancheng Lou \at
Google Inc., California, USA \\
\email{tiancheng.lou@gmail.com}
\and
Dongxiao Yu \at Department of Computer Science, the University of Hong
Kong, Hong Kong, China\\
%\email{}
}
\maketitle

\begin{abstract}
Cognitive Radio Networks (CRNs) are considered as a promising
solution to the spectrum shortage problem in wireless communication.
In this paper, we initiate the first systematic study on the algorithmic complexity of the
connectivity problem in CRNs through spectrum assignments. We model
the network of secondary users (SUs) as a potential graph, where two nodes having an edge between them are connected as long as
they choose a common available channel. In the general case, where
the potential graph is arbitrary and the SUs may have different number
of antennae, we prove that it is NP-complete to determine whether
the network is connectable even if there are only two channels. For the
special case where the number of channels is constant and all the SUs
have the same number of antennae, which is more than one but less
than the number of channels, the problem is also NP-complete. For
the special cases in which the potential graph is complete, a tree, or a graph with bounded treewidth, we
prove the problem is NP-complete and fixed-parameter tractable (FPT)
when parameterized by the number of channels. Exact
algorithms are also derived to determine the connectability of a given cognitive radio network.
\end{abstract}

\section{Introduction}

Cognitive Radio is a promising technology to alleviate the spectrum
shortage in wireless communication. It allows the unlicensed
\emph{secondary users} to utilize the temporarily unused licensed
spectrums, referred to as \emph{white spaces}, without interfering
with the licensed \emph{primary users}. Cognitive Radio Networks
(CRNs) is considered as the next generation of communication
networks and attracts numerous research from both academia and
industry recently.

In CRNs, each secondary user (SU) can be equipped with one or
multiple antennae for communication. With multiple antennae, a SU
can communicate on multiple channels simultaneously (in this paper,
channel and spectrum are used interchangeably.). Through spectrum
sensing, each SU has the capacity to measure current available
channels at its site, i.e. the channels are not used by the primary
users (PUs). Due to the appearance of PUs, the available channels of
SUs have the following characteristics~\cite{sigcomm09Bahl}:
\begin{itemize}
\item \emph{Spatial Variation}: SUs at different positions may have
different available channels;
\item \emph{Spectrum Fragmentation}: the available channels of a SU may not
be continuous; and
\item \emph{Temporal Variation}: the available channels of a SU may change
over time.
\end{itemize}

Spectrum assignment is to allocate available channels to SUs to
improve system performance such as spectrum utilization, network
throughput and fairness. Spectrum assignment is one of the most
challenging problems in CRNs and has been extensively
studied such as in~\cite{infocom12Li,suveryWang,mobihoc12Huang,mobihoc07Yuan,auctionZhou}.

Connectivity is a fundamental problem in wireless communication.
Connection between two nodes in CRNs is not only determined by their
distance and their transmission powers, but also related to whether
the two nodes has chosen a common channel. Due to the spectrum
dynamics, communication in CRNs is more difficult than in the
traditional multi-channel radio networks studied in~\cite{discDolev}. Authors
in ~\cite{Infocom12Lu,CoRoNetRen,JsacRen} investigated the impact of
different parameters on connectivity in large-scale CRNs, such as
the number of channels, the activity of PUs, the number of neighbors
of SUs and the transmission power.

\begin{figure}[t]%htbp
\begin{minipage}[b]{0.5\textwidth}
\centering
        \includegraphics[width=0.9\textwidth]{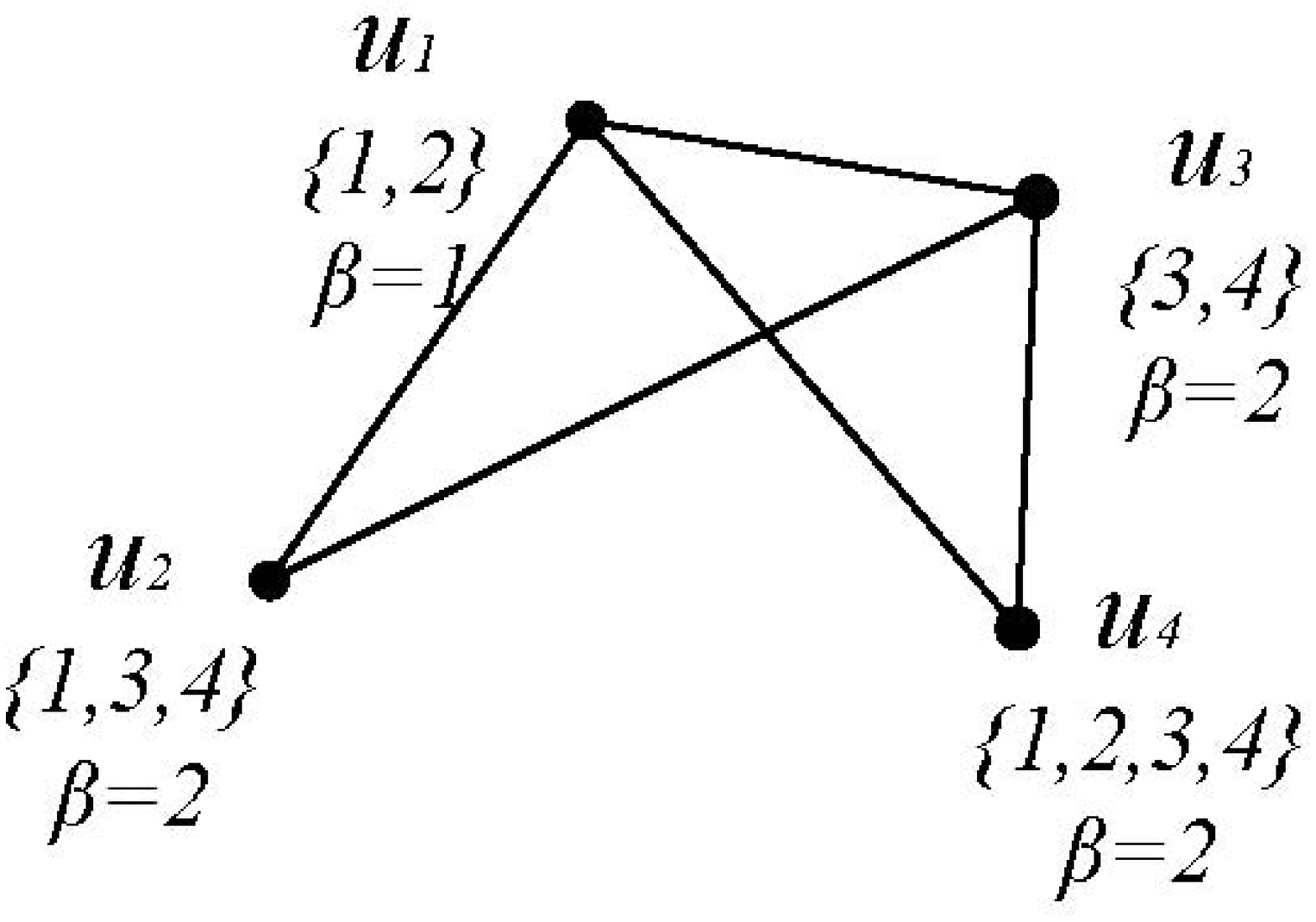}     %\label{fig:general:a}
\end{minipage}
\begin{minipage}[b]{0.5\textwidth}
\centering
        \includegraphics[width=0.9\textwidth]{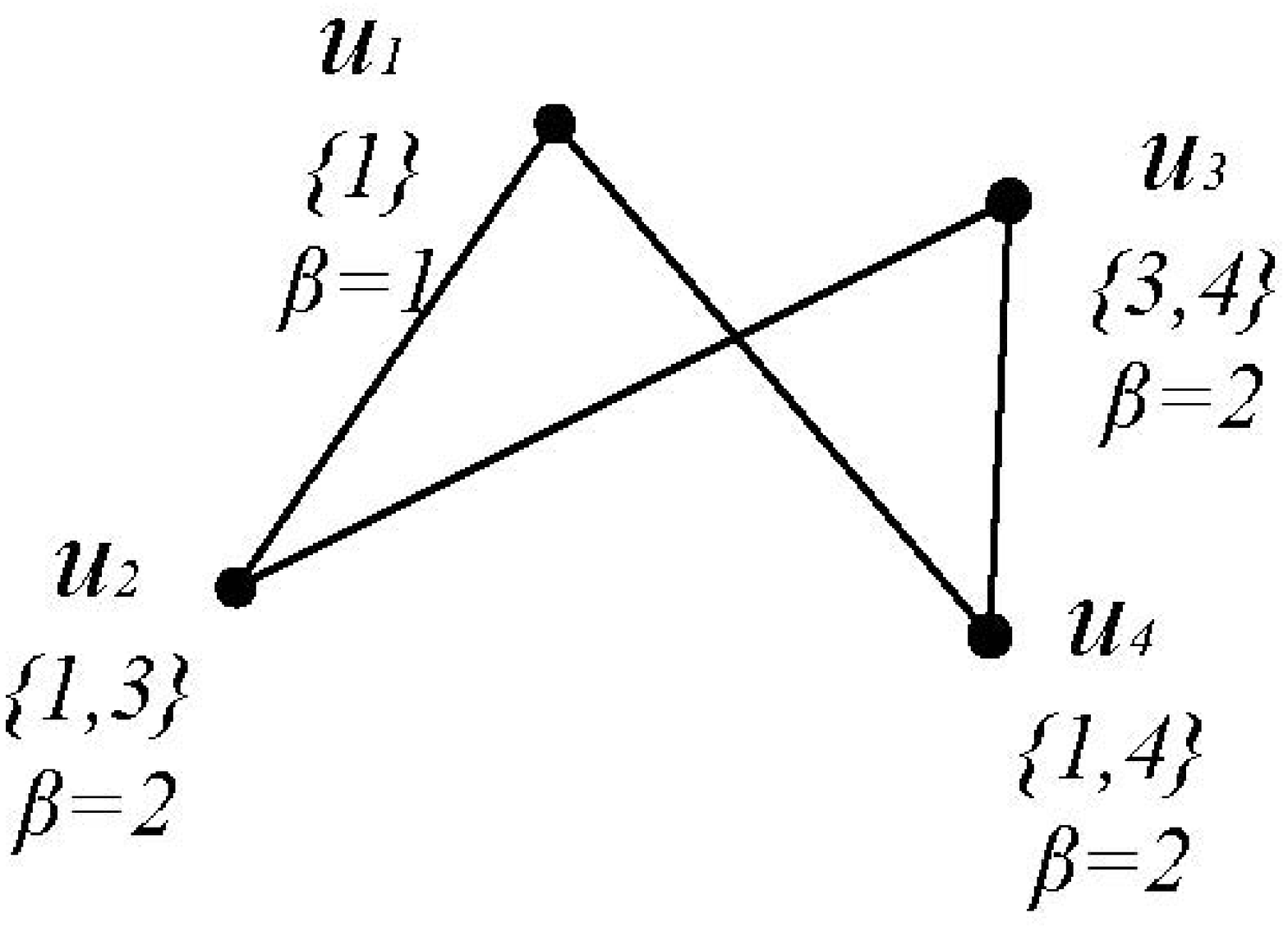}
%            \label{fig:genreal:b}
\end{minipage}
 \caption{the general case. a) the potential graph: the set besides each SU is its available channels, and $\beta$ is its number of antennae.  $u_2$ and $u_4$ are not connected directly because they are a pair of heterogenous nodes or their distance exceeds at least one of their transmission ranges. b) the realization graph which is connected: the set beside each SU is the channels assigned to it. }
\label{fig:general}
\end{figure}

In this paper, we initiate the first systematic study on the complexity of connectivity in CRNs
through spectrum assignment. We model the network as a potential
graph and a realized graph before and after spectrum assignment
respectively (refer to Section~\ref{sec:def}).  We start from the
most general case, where the network is composed of heterogenous
SUs\footnote{We assume two heterogenous SUs cannot communicate even
when they work on a common channel and their distance is within
their transmission ranges.}, SUs may be equipped with different number
of antennae and the potential graph can be arbitrary
(Figure~\ref{fig:general}). Then, we proceed to study the special
case when all the SUs have the same number of antennae. If all the
SUs are homogenous with transmission ranges large enough, the
potential graph will be a complete graph. For some hierarchically
organized networks, e.g. a set of SUs are connected to an access
point, the potential graph can be a tree. Therefore, we also study
these special cases. Exact algorithms are also derived to determine
connectivity for different cases. Our results are listed below. To
the best of knowledge, this is the first work that systematically
studies the algorithmic complexity of connectivity in CRNs with
multiple antennae.

\paragraph{Our Contributions:} In this paper we study the algorithmic complexity of the  connectivity problem through spectrum assignment under different models. Our main results are as follows.
\begin{itemize}
\item When the potential graph is a general graph, we prove that the problem is NP-complete even if there are only two channels. This result is sharp as the problem is polynomial-time
solvable when there is only one channel. We also design exact
algorithms for the problem. For the special case when all SUs have
the same number of antennae, we prove that the problem is
NP-complete when $k>\beta\geq 2$, where $k$ and $\beta$ are the
total amount of channels in the white spaces and the number of
antennae on an SU respectively.

\item When the potential graph is complete,\footnote{The complete graph is a special case of disk graphs, which are commonly used to model wireless networks such as in~\cite{wireless08Kuhn,jsac12Du}.} the problem is shown to be NP-complete even if each node can open at most two channels. However, in contrast to the general case, the problem is shown to be polynomial-time solvable if the number of channels is fixed. In fact, we prove a stronger result saying that the problem is fixed parameter tractable when parameterized by the number of channels.
    (See \cite{book_para} for notations in parameterized complexity.)

\item When the potential graph is a tree, we prove that the problem is NP-complete even if the tree has depth one. Similar to the complete graph case, we show that the problem is fixed parameter tractable when parameterized by the number of channels. We then generalize this result, showing that the problem remains fixed parameter tractable when parameterized by the number of channels if the underlying potential graph has bounded treewidth.
\end{itemize}

\paragraph{Paper Organization:} In Section~\ref{sec:def} we formally define our model and problems studied in this paper. We study the problem with arbitrary potential graphs in Section~\ref{sec:general}. The special cases where the potential graph is complete or a tree are investigated in Sections~\ref{sec:complete} and \ref{sec:trees}. The paper is concluded in Section~\ref{sec:con} with possible future works.

\section{Preliminaries}\label{sec:def}
\subsection{System Model and Problem Definition}\label{subsec:def}
We first describe the model used throughout this paper. A
\emph{\network~}is comprised of the following ingredients:
\begin{itemize}
\item $U$ is a collection of secondary users (SUs) and $C$ is the set of channels in the white spaces.
\item Each SU $u\in U$ has a \emph{\specm,} denoted by \sm($u$), which is a subset of $C$ representing the available channels that $u$ can open.
\item The \emph{\pgraph~}$\mathcal{PG}=(U,E)$, where each edge of $E$ is also called a
\emph{potential edge}. If two nodes are connected by a potential
edge, they can communicate as long as they choose a common available
channel.
\item Each SU $u \in U$ is equipped with a number of antennas, denoted as \emph{\budget~}$\beta(u)$, which is the maximum number of channels that $u$ can open simultaneously.
\end{itemize}

For a set $S$, let $2^{S}$ denote the power set of $S$, i.e., the collection of all subsets of $S$.
A \emph{\specassign~}is a function $\mathcal{SA}: U \rightarrow 2^{C}$ satisfying that
$$\mathcal{SA}(u) \subseteq \sm(u) \textrm{~and~} |\mathcal{SA}(u)| \leq \beta(u) \textrm{~for all~} u\in U.$$
Equivalently, a \specassign is a way of SUs opening channels such
that each SU opens at most $\beta$ channels and can only open those in its \specm.

Given a \specassign $\mathcal{SA}$, a \pedge $\{u,v\}\in E$ is called \emph{realized~}if $\mathcal{SA}(u)\cap \mathcal{SA}(v)\neq \emptyset$, i.e., there exists a channel opened by both $u$ and $v$. The \emph{\realgraph~}under a \specassign is a graph $\mathcal{RG}=(U,E')$, where $E'$ is the set of realized edges in $E$. Note that $\mathcal{RG}$ is a spanning subgraph of the \pgraph $\mathcal{PG}$. A \network is called \emph{connectable~}if there exists a \specassign under which the \realgraph is connected, in which case we also say that the \network is \emph{connected} under this \specassign.
Now we can formalize the problems studied in this paper.

\vspace{1.5mm}
\noindent\textbf{The \fullprob Problem}.
The \fullprob problem is to decide whether a given \network is connectable.

\vspace{1.5mm} We are also interested in the special case where the
number of possible channels is small\footnote{Commonly, the white
spaces include spectrums from channel 21 (512Mhz) to 51 (698Mhz)
excluding channel 37, which is totally 29
channels~\cite{sigcomm09Bahl}.} and SUs have the same antenna
budget. Therefore, we define the following subproblem of the
\fullprob problem:

\vspace{1.5mm} \noindent\textbf{The Spectrum
$(k,\beta)$-Connectivity Problem}. For two constants $k,\beta\geq
1$, the \textsc{Spectrum} $(k,\beta)$-\textsc{Connectivity} problem
is to decide whether a given \network with $k$ channels in which all
SUs have the same budget $\beta$ is connectable. For convenience we
write \prob($k,\beta$) to represent this problem.

\vspace{1.5mm} Finally, we also consider the problem with special
kinds of \pgraphs, i.e. the potential graph is complete or a tree.

In the sequel, unless otherwise stated, we always use $n:=|U|$ and $k:=|C|$ to denote respectively the number of secondary users and channels.

\subsection{Tree Decomposition}\label{subsec:treedecomp}
In this subsection we give some basic notions regarding the tree decomposition of a graph, which will be used later.
The concept of treewidth was introduced by Robertson and Seymour in their seminal work on graph minors \cite{rs_ja86}. A {\em tree decomposition} of a graph $G=(V,E)$ is given by a tuple
$(T=(I,F), \{X_i~|~i\in I\})$, where $T$ is a tree and each $X_i$ is a subset
of $V$ called a {\em bag} satisfying that
\begin{itemize}
\item $\bigcup_{i\in I} X_i = V$;
\item For each edge $\{u,v\}\in E$, there exists a tree node $i$ with
  $\{u,v\} \subseteq X_i$;
\item For each vertex $u\in V$, the set of tree nodes $\{i\in I~|~u\in
  X_i\}$ forms a connected subtree of $T$. Equivalently, for any three
  vertices $t_1,t_2,t_3 \in I$ such that $t_2$ lies in the path
  from $t_1$ to $t_3$, it holds that $X_{t_1} \cap X_{t_3} \subseteq
  X_{t_2}$.
\end{itemize}

The {\em width} of the tree decomposition is $\max_{i\in I}
\{|X_i|-1\}$, and the {\em treewidth} of a graph $G$ is the
minimum width of a tree decomposition of $G$. For each fixed integer $d$, there is a polynomial time algorithm that decides whether a given graph has treewidth at most $d$, and if so, constructs a tree decomposition of width $d$ \cite{bod_sjc96}. Such a decomposition can easily be transformed to a \emph{nice tree decomposition} $(T,\{X_i\})$ of $G$ with the same width, in which $T$ is a rooted binary tree with at most $O(|V|)$ nodes (see e.g. \cite{kloks94}).

\section{The \fullprob Problem}\label{sec:general}
In this section, we study the the \fullprob problem from both
complexity and algorithmic points of view.

\subsection{NP-completeness Results}
We show that the \fullprob problem is NP-complete even if the number of channels is fixed. In fact we give a complete characterization of the complexity of \prob($k,\beta$) by proving the following dichotomy result:

\begin{theorem}\label{thm:general}
\prob($k,\beta$) is \NP-complete for any integers $k>\beta\geq 2$,
and is in \P if $\beta=1$ or $k\leq \beta$.
\end{theorem}

The second part of the statement is easy: When $\beta=1$, each SU
can only open one channel, and thus all SUs should be connected
through the same channel. Therefore, the network is connectable if
and only if there there exists a channel that belongs to every SU's
\specm (and of course the \pgraph must be connected), which is easy
to check. When $k\leq \beta$, each SU can open
all channels in its \specm, and the problem degenerates to checking
the connectivity of the \pgraph.

In the sequel we prove the \NP-completeness of \prob($k,\beta$) when
$k>\beta\geq 2$. First consider the case $k=\beta+1$. We will reduce
a special case of the Boolean Satisfiability (SAT) problem, which will be
shown to be \NP-complete, to \prob($\beta+1,\beta$), thus showing
the \NP-completeness of the latter.

A clause is called \emph{positive} if it only contains positive
literals, and is called \emph{negative} if it only contains negative
literals. For example, $x_1 \lor x_3 \lor x_5$ is positive and
$\overline{x_2} \lor \overline{x_4}$ is negative. A clause is called
\emph{uniform} if it is positive or negative. A \emph{uniform} CNF formula is the conjunction of uniform clauses. Define \unisat as the
problem of deciding whether a given uniform CNF formula is
satisfiable.

\begin{lemma}\label{lem:unisat}
\unisat is \NP-complete.
\end{lemma}
\begin{proof}
Let $F$ be a CNF formula with variable set $\{x_1,x_2,\ldots,x_n\}$. For each $i$ such that $\overline{x_i}$ appears in $F$, we create a new variable $y_i$, and do the following:
\begin{itemize}
\item substitute $y_i$ for all occurrences of $\overline{x_i}$;
\item add two clauses $x_i \lor y_i$ and $\overline{x_i} \lor \overline{y_i}$ to $F$. More formally, let $F \leftarrow F \land (x_i \lor y_i) \land (\overline{x_i} \lor \overline{y_i})$. This ensures $y_i = \overline{x_i}$ in any satisfying assignment of $F$.
\end{itemize}
Call the new formula $F'$.
For example, if $F=(x_1\lor \overline{x_2}) \land (\overline{x_1} \lor x_3)$, then $F'=(x_1\lor y_2) \land (y_1 \lor x_3) \land (x_1 \lor y_1) \land (\overline{x_1} \lor \overline{y_1}) \land (x_2 \lor y_2) \land (\overline{x_2} \lor \overline{y_2})$.

It is easy to see that $F'$ is a uniform CNF formula, and that $F$ is satisfiable if and only if $F'$ is satisfiable. This constitutes a reduction from \textsc{SAT} to \unisat, which concludes the proof. \qed
\end{proof}

\begin{theorem}\label{thm:npc}
\prob($\beta+1,\beta$) is \NP-complete for any integer $\beta\geq
2$.
\end{theorem}
\begin{proof}
The membership of \prob($\beta+1,\beta$) in \NP is clear. In what
follows we reduce \unisat to \prob($\beta+1,\beta$), which by
Lemma~\ref{lem:unisat} will prove the \NP-completeness of the
latter.

Let $c_1\land c_2\land \ldots \land c_m$ be an input to \unisat
where $c_j$, $1\leq j\leq m$, is a uniform clause. Assume the variable set is $\{x_1,x_2,\ldots,x_n\}$.
We construct an instance of \prob($\beta+1,\beta$) as follows.

\begin{itemize}
\item \textbf{Channels:} There are $\beta+1$ channels $\{0,1,2,\ldots,\beta\}$.

\item \textbf{SUs:}
%\begin{itemize}
\begin{itemize}
\item For each variable $x_i$, there is a corresponding SU $X_i$
with spectrum map $\sm(X_i)=\{0,1,2,\ldots,\beta\}$ (which contains all
possible channels);

\item for each clause $c_j$, $1\leq j\leq m$, there is a corresponding
SU $C_j$ with $\sm(C_j)=\{p_j\}$, where $p_j=1$ if $c_j$ is positive
and $p_j=0$ if $c_j$ is negative;

\item there is an SU $Y_{2}$ with $\sm(Y_2)=\{2\}$. For every $1\leq i\leq n$ and $2\leq k\leq \beta$, there is an SU $Y_{i,k}$ with $\sm(Y_{i,k})=\{k\}$; and
\item all SUs have the same \budget $\beta$.
\end{itemize}

\item \textbf{Potential Graph:} For each clause $c_j$ and each
variable $x_i$ that appears in $c_j$ (either as $x_i$ or
$\overline{x_i}$), there is a \pedge between $X_i$ and
$C_j$. For each $1\leq i\leq n$ and $3\leq k\leq \beta$, there is a
\pedge between $X_i$ and $Y_{i,k}$. Finally, there is a
\pedge between $Y_2$ and every $X_i$, $1\leq i\leq n$.
\end{itemize}

Denote the above \network by $\mathcal{I}$, which is also an instance of \prob($\beta+1,\beta$). We now prove that
$c_1\land c_2\land \ldots \land c_m$ is satisfiable if and only if
$\mathcal{I}$ is connectable.

First consider the ``only if'' direction. Let
$A:\{x_1,\ldots,x_n\}\rightarrow\{0,1\}$ be a satisfying assignment
of $c_1\land c_2\land \ldots \land c_m$, where $0$ stands for FALSE
and $1$ for TRUE. Define a \specassign as follows. For each $1\leq
i\leq n$, let user $X_i$ open the channels $\{2,3,\ldots,\beta\}
\cup \{A(i)\}$. Every other SU opens the only channel in its \specm.

We verify that $\mathcal{I}$ is connected under the above
\specassign. For each $1\leq i\leq n$, $X_i$ is connected to $Y_2$
through channel 2. Then, for every $2\leq l\leq \beta$, $Y_{i,l}$ is
connected to $X_i$ through channel $l$. Now consider SU $C_j$ where
$1\leq j\leq m$. Since $A$ satisfies the clause $c_j$, there exists
$1\leq i\leq n$ such that:
1) $x_i$ or $\overline{x_i}$ occurs in $c_j$; and
2) $A(x_i)=1$ if $c_j$ is positive, and $A(x_i)=0$ if $c_j$ is
negative.
Thus $X_i$ and $C_j$ are connected through channel $A(x_i)$.
Therefore the \realgraph is connected, completing
the proof of the ``only if'' direction.

We next consider the ``if'' direction. Suppose there is a \specassign that makes $\mathcal{I}$ connected. For every $1\leq i\leq n$ and $2\leq l\leq \beta$, $X_i$ must open channel $l$, otherwise $Y_{i,l}$ will become an isolated vertex in the \realgraph. Since $X_i$ can open at most $\beta$ channels in total, it can open at most one of the two remaining channels $\{0,1\}$. We assume w.l.o.g. that $X_i$ opens exactly one of them, which we denote by $a_i$.

Now, for the formula $c_1\land c_2\land \ldots \land c_m$, we define
a truth assignment $A:\{x_1,\ldots,x_n\}\rightarrow\{0,1\}$ as
$A(x_i)=a_i$ for all $1\leq i\leq n$. We show that $A$ satisfies the
formula. Fix $1\leq j\leq m$ and assume that $c_j$ is negative (the
case where $c_j$ is positive is totally similar). Since the \specm
of SU $C_j$ only contains channel $0$, some of its neighbors must
open channel 0. Hence, there exists $1\leq i\leq n$ such that
$\overline{x_i}$ appears in $c_j$ and the corresponding SU $X_i$
opens channel 0. By our construction of $A$, we have $A(x_i)=0$, and
thus the clause $c_j$ is satisfied by $A$. Since $j$ is chosen
arbitrarily, the formula $c_1\land c_2\land \ldots \land c_m$ is
satisfied by $A$. This completes the reduction from \unisat to
\prob($\beta,\beta+1$), and the theorem follows. \qed
\end{proof}

\begin{corollary}\label{cor:npc}
\prob($k,\beta$) is \NP-complete for any integers $k>\beta\geq
2$.
\end{corollary}
\begin{proof}
By a simple reduction from \prob($\beta+1,\beta$): Given an instance
of \prob($\beta+1,\beta$), create $k-\beta-1$ new channels and add
them to the \specm of an (arbitrary) SU. This gives a instance of
\prob($k,\beta$). Since the new channels are only contained in one
SU, they should not be opened, and thus the two instances are
equivalent. Hence the theorem follows. \qed
\end{proof}

Theorem~\ref{thm:npc} indicates that the \fullprob problem is \NP-complete even if the \network only has three channels. We further strengthen this result by proving the following theorem:

\begin{theorem}\label{thm:npc_2channels}
The \fullprob problem is \NP-complete even if there are only two channels.
\end{theorem}
\begin{proof}
We present a reduction from \unisat similar as in the proof of
Theorem~\ref{thm:npc}. Let $c_1\land c_2\land \ldots \land c_m$ be a
uniform CNF clause with variable set $\{x_1,x_2,\ldots,x_n\}$.
Construct a \network as follows: There are two channels \{0,1\}. For
each variable $x_i$ there is a corresponding SU $X_i$ with \specm
$\sm(X_i)=\{0,1\}$ and \budget $\beta(X_i)=1$. For each clause $c_j$
there is a corresponding SU $C_j$ with $\sm(C_j)=\{p_j\}$ and
$\beta(C_j)=1$, where $p_j=1$ if $c_j$ is positive and $p_j=0$ if
$c_j$ is negative. There is an SU $Y$ with $\sm(Y)=\{0,1\}$ and
$\beta(Y)=2$. Note that, unlike in the case of \prob($k,\beta$), SUs
can have different antenna budgets. Finally, the edges of the
\pgraph include: $\{X_i,C_j\}$ for all $i,j$ such that $x_i$ or
$\overline{x_i}$ appears in $c_j$, and $\{Y,X_i\}$ for all $i$. This
completes the construction of the \network, which is denoted by
$\mathcal{I}$. By an analogous argument as in the proof of
Theorem~\ref{thm:npc}, $c_1\land c_2\land \ldots \land c_m$ is
satisfiable if and only if $\mathcal{I}$ is connectable, concluding
the proof of Theorem~\ref{thm:npc_2channels}. \qed
\end{proof}

Theorem~\ref{thm:npc_2channels} is sharp in that, as noted before, the problem is polynomial-time solvable when there is only one channel.

\subsection{Exact Algorithms}
In this subsection we design algorithms for deciding whether a given \network is connectable.
Since the problem is NP-complete, we cannot expect a polynomial time algorithm.

Let $n,k,t$ denote the number of SUs, the number of channels, and
the maximum size of any SU's \specm, respectively ($t\leq k$). The
simplest idea is to exhaustively examine all possible spectrum
assignments to see if there exists one that makes the network
connected. Since each SU can have at most $2^{t}$ possible ways of
opening channels, the number of assignments is at most $2^{tn}$.
Checking each assignment takes poly($n,k$) time. Thus the running
time of this approach is bounded by $2^{tn}(nk)^{O(1)}$, which is
reasonable when $t$ is small. However, since in general $t$ can be
as large as $k$, this only gives a $2^{O(kn)}$ bound, which is
unsatisfactory if $k$ is large. In the following we present another
algorithm for the problem that runs faster than the above approach
when $k$ is large.

\begin{theorem}\label{thm:alg_general}
There is an algorithm that decides whether a given \network is
connectable in time $2^{O(k+n\log n)}$.
\end{theorem}
\begin{proof}
Let $\mathcal{I}$ be a given \network with \pgraph $\mathcal{PG}$.
Let $n$ be the number of SUs and $k$ the number of channels. Assume
that $\mathcal{I}$ is connected under some \specassign. Clearly the
realization graph contains a spanning tree of $\mathcal{PG}$, say
$T$, as a subgraph. If we change the \pgraph to $T$ while keeping
all other parameters unchanged, the resulting network will still be
connected under the same \specassign. Thus, it suffices to check
whether there exists a spanning tree $T$ of $\mathcal{G}$ such that
$\mathcal{I}$ is connectable when substituting $T$ for
$\mathcal{PG}$ as its \pgraph. Using the algorithm of
\cite{spanntree}, we can list all spanning trees of $\mathcal{PG}$
in time $O(Nn)$ where $N$ is the number of spanning trees of
$\mathcal{PG}$. By Cayley's formula \cite{cayley,cayley_arxiv} we
have $N\leq n^{n-2}$. Finally, for each spanning tree $T$, we can
use the algorithm in Theorem~\ref{thm:alg_tree} (which will appear
in Section~\ref{sec:trees}) to decide whether the network is
connectable in time $2^{O(k)}n^{O(1)}$. The total running time of
the algorithm is $O(n^{n-2})2^{O(k)}n^{O(1)}=2^{O(k+n\log n)}$.
\qed
\end{proof}

Combining Theorem~\ref{thm:alg_general} with the brute-force approach, we obtain:

\begin{corollary}\label{cor:alg_general}
The \fullprob problem is solvable can be solved in time $2^{O(\min\{kn,k+n\log
n\})}$.
\end{corollary}

\section{\fullprob with Complete Potential
Graphs}\label{sec:complete}
In this section we consider the special case of the \fullprob
problem, in which the \pgraph of the \network is complete. We first
show that this restriction does not make the problem tractable in
polynomial time.

\begin{theorem}\label{thm:npc_complete}
The \fullprob problem is NP-complete even when the \pgraph is
complete and all SUs have the same \budget $\beta=2$.
\end{theorem}
\begin{proof}
The membership in \NP is trivial. The hardness proof is by a
reduction from the \textsc{Hamiltonian Path} problem, which is to
decide whether a given graph contains a Hamiltonian path, i.e., a
simple path that passes every vertex exactly once. The
\textsc{Hamiltonian Path} problem is well-known to be NP-complete
\cite{book_npc}. Let $G=(V,E)$ be an input graph of the
\textsc{Hamiltonian Path} problem. Construct an instance of the
\fullprob problem as follows: The collection of channels is $E$ and
the set of SUs is $V$; that is, we identify a vertex in $V$ as an SU
and an edge in $E$ as a channel. For every $v\in V$, the \specm of
$v$ is the set of edges incident to $v$. All SUs have \budget
$\beta=2$. Denote this \network by $\mathcal{I}$. We will prove that
$G$ contains a Hamiltonian path if and only if $\mathcal{I}$ is
connectable.

First suppose $G$ contains a Hamiltonian path $P=v_1v_2\ldots v_n$,
where $n=|V|$. Consider the following \specassign of $\mathcal{I}$:
for each $1\leq i\leq n$, let SU $v_i$ open the channels
corresponding to the edges incident to $v_i$ in the path $P$. Thus
all SUs open two channels except for $v_1$ and $v_n$ each of whom
opens only one. For every $1\leq i\leq n-1$, $v_i$ and $v_{i+1}$ are
connected through the channel (edge) $\{v_i,v_{i+1}\}$. Hence the
realization graph of $\mathcal{I}$ under this \specassign is
connected.

Now we prove the other direction. Assume that $\mathcal{I}$ is
connectable. Fix a \specassign under which the \realgraph of
$\mathcal{I}$ is connected, and consider this particular \realgraph
$\mathcal{RG}=(V,E')$. Let $\{v_i,v_j\}$ be an arbitrary edge in
$E'$. By the definition of the \realgraph, there is a channel opened
by both $v_i$ and $v_j$. Thus there is an edge in $E$ incident to
both $v_i$ and $v_j$, which can only be $\{v_i,v_j\}$. Therefore
$\{v_i,v_j\}\in E$. This indicates $E'\subseteq E$, and hence
$\mathcal{RG}$ is a connected spanning subgraph of $G$. Since each
SU can open at most two channels, the maximum degree of
$\mathcal{RG}$ is at most 2. Therefore $\mathcal{RG}$ is either a
Hamiltonian path of $G$, or a Hamiltonian cycle which contains a
Hamiltonian path of $G$. Thus, $G$ contains a Hamiltonian path.

The reduction is complete and the theorem follows. \qed
\end{proof}

Notice that the reduction used in the proof of Theorem~\ref{thm:npc_complete} creates a \network with an unbounded number of channels. Thus Theorem~\ref{thm:npc_complete} is not stronger than Theorem~\ref{thm:general} or \ref{thm:npc_2channels}.
Recall that Theorem \ref{thm:npc_2channels} says the \fullprob problem is \NP-complete even if there are only two channels.
In contrast we will show that, with complete \pgraphs, the problem is polynomial-time tractable when the number of channels is small.

\begin{theorem}\label{thm:alg_complete}
The \fullprob problem with complete \pgraphs can be solved in
$2^{2^k+O(k)}n^{O(1)}$ time.
%Thus, the problem with complete \pgraphs is fixed parameter tractable (FPT) when parameterized by the number of channels.
\end{theorem}
\begin{proof}
Consider a \network $\mathcal{I}$ with SU set $U$, channel set $C$
and a complete \pgraph, i.e., there is a potential edge between
every pair of distinct SUs. Recall that $n=|U|$ and $k=|C|$. For each
\specassign $\mathcal{SA}$, we construct a corresponding
\emph{spectrum graph} $\mathcal{G}_{\mathcal{SA}}=(V, E)$ where
$V=\{C'\subseteq C~|~\exists u\in U \textrm{~s.t.~}
\mathcal{SA}(u)=C'\}$ and $E=\{\{C_1,C_2\}~|~C_1,C_2\in V; C_1 \cap
C_2 \neq \emptyset\}$. Thus, $V$ is the collection of subsets of $C$
that is opened by some SU, and $E$ reflexes the connectivity between
pairs of SUs that open the corresponding channels. Since each vertex
in $V$ is a subset of $C$, we have $|V|\leq 2^{k}$, and the number
of different spectrum graphs is at most $2^{2^{k}}$.

We now present a relation between $\mathcal{G}_{\mathcal{SA}}$ and the realization graph of $\mathcal{I}$ under $\mathcal{SA}$. If we label each vertex $u$ in the realization graph with $\mathcal{SA}(u)$, and contract all edges between vertices with the same label, then we obtain precisely the spectrum graph $\mathcal{G}_{\mathcal{SA}}=(V, E)$. Therefore, in the language of graph theory, $\mathcal{G}_{\mathcal{SA}}=(V, E)$ is a minor of the realization graph under $\mathcal{SA}$. Since graph minor preserves connectivity, $\mathcal{I}$ is connectable if and only if there exists a connected spectrum graph. Hence we can focus on the problem of deciding whether a connected spectrum graph exists.

Consider all possible graphs $G=(V,E)$ such that $V\subseteq 2^{C}$,
and $E=\{\{C_1,C_2\}~|~C_1,C_2\in V; C_1 \cap C_2 \neq \emptyset\}$.
There are $2^{2^k}$ such graphs each of which has size $2^{O(k)}$.
Thus we can list all such graphs in $2^{2^k+O(k)}$ time. For each
graph $G$, we need to check whether it is the spectrum graph of some
\specassign of $\mathcal{I}$. We create a bipartite graph in which
nodes on the left side are the SUs in $\mathcal{I}$, and nodes on
the right side all the vertices of $G$. We add an edge between an SU
$u$ and a vertex $C'$ of $G$ if and only if $C' \subseteq \sm(u)$
and $|C'| \leq \beta(u)$, that is, $u$ can open $C'$ in a
\specassign. The size of $H$ is poly($n,2^k$) and its construction
can be finished in poly($n,2^k$) time. Now, if $G$ is the spectrum
graph of some \specassign $\mathcal{SA}$, then we can identify
$\mathcal{SA}$ with a subgraph of $H$ consisting of all edges $(u,
\mathcal{SA}(u))$ where $u$ is an SU. In addition, in this subgraph
we have
\begin{itemize}
\item every SU $u$ has degree exactly one; and
\item every node $C'$ on the right side of $H$ has degree at least one.
\end{itemize}

Conversely, a subgraph of $H$ satisfying the above two conditions clearly induces a \specassign whose spectrum graph is exactly $G$. Therefore it suffices to examine whether $H$ contains such a subgraph. Furthermore, the above conditions are easily seen to be equivalent to:

\begin{itemize}
\item every SU $u$ has degree at least one in $G$; and
\item $G$ contains a \emph{matching} that includes all nodes on the right side.
\end{itemize}

The first condition can be checked in time linear in the size of $H$, and the second one can be examined by any polynomial time algorithm for bipartite matching (e.g., \cite{bmatching}). Therefore, we can decide whether such subgraph exists (and find one if so) in time poly($n,2^k$). By our previous analyses, this solves the \fullprob problem with complete potential graphs. The total running time of our algorithm is $2^{2^k+O(k)}{\rm poly}(n,2^k)=2^{2^k+O(k)}n^{O(1)}$. \qed
\end{proof}

%The following corollary immediate follows from Theorem~\ref{thm:alg_complete}.

\begin{theorem}\label{cor:fpt_complete}
The \fullprob problem with complete \pgraphs is fixed parameter tractable (FPT) when parameterized by the number of channels.
\end{theorem}

\section{Spectrum Connectivity on Trees and Bounded Treewidth Graphs}\label{sec:trees}
In this section, we study another special case of the \fullprob problem where the \pgraph of the \network is a tree. We will also investigate the problem on the class of bounded-treewidth graphs.
Many \NP-hard combinatorial problems become easy on trees, e.g., the dominating set problem and the vertex cover problem. Nonetheless, as indicated by the following theorem, the \fullprob problem remains hard on trees.

\subsection{Trees}

We state the complexity of the spectrum connectivity problem with trees as the \pgraph in the following theorem.

\begin{theorem}\label{thm:npc_tree}
The \fullprob problem is \NP-complete even if the \pgraph is a tree of depth one.
\end{theorem}
\begin{proof}
We give a reduction from the \textsc{Vertex Cover} problem which is
well known to be \NP-complete \cite{book_npc}. Given a graph
$G=(V,E)$ and an integer $r$, the \textsc{Vertex Cover} problem is
to decide whether there exists $r$ vertices in $V$ that cover all
the edges in $E$. Construct a \network $\mathcal{I}$ as follows. The
set of channels is $C=\{c_v~|~v\in V\}$. For each edge $e=\{u,v\}\in
E$ there is an SU $U_{e}$ with $\sm(U_{e})=\{c_u,c_v\}$ and \budget
2. There is another SU $M$ with $\sm(M)=C$ and \budget $r$. The
\pgraph is a star centered at $M$, that is, there is a \pedge
between $M$ and $U_{e}$ for every $e\in E$. This finishes the
construction of $\mathcal{I}$.

We prove that $G$ has a vertex cover of size $r$ if and only if $\mathcal{I}$ is connectable.
First assume $G$ has a vertex cover $S \subseteq V$ with $|S|\leq r$. Define a \specassign $A(S)$ as follows: $M$ opens the channels $\{c_v~|~v\in S\}$, and $U_e$ opens both channels in its \specm for all $e\in E$.
Since $S$ is a vertex cover, we have $u\in S$ or $v\in S$ for each $e=\{u,v\}\in E$. Thus at least one of $c_u$ and $c_v$ is opened by $M$, which makes it connected to $U_e$. Hence the \realgraph is connected. On the other hand, assume that the \realgraph is connected under some \specassign. For each $e=\{u,v\}\in E$, since the \pedge $\{M,U_{e}\}$ is realized, $M$ opens at least one of $c_u$ and $c_v$. Now define $S=\{v\in V~|~c_v \textrm{~is opened by~}M\}$. It is clear that $S$ is a vertex cover of $G$ of size at most $\beta(M)=r$. This completes the reduction, and the theorem follows. \qed
\end{proof}

We next show that, in contrast to Theorems~\ref{thm:npc} and \ref{thm:npc_2channels}, this special case of the problem is polynomial-time solvable when the number of channels is small.

\begin{theorem}\label{thm:alg_tree}
Given a \network whose \pgraph is a tree, we can check whether it is
connectable in $2^{O(t)}(kn)^{O(1)}$ time, where $t$ is the maximum size
of any SU's \specm. In particular, this running time is at most
$2^{O(k)}n^{O(1)}$.
\end{theorem}
\begin{proof}
Let $\mathcal{I}$ be a given \network whose \pgraph $\mathcal{PG}=(V,E)$ is a tree. Root $\mathcal{PG}$ at an arbitrary node, say $r$. For each $v\in V$ let $\mathcal{PG}_v$ denote the subtree rooted at $v$, and let $\mathcal{I}_v$ denote the \network obtained by restricting $\mathcal{I}$ on $\mathcal{PG}_v$.  For every subset $S\subseteq \sm(v)$, define $f(v,S)$ to be 1 if there exists a \specassign that makes $\mathcal{I}_v$ connected in which the set of channels opened by $v$ is exactly $S$; let $f(v,S)=0$ otherwise. For each channel $c\in C$, define $g(v,c)$ to be 1 if there exists $S$, $\{c\}\subseteq S \subseteq \sm(v)$, for which $f(v,S)=1$; define $g(v,c)=0$ otherwise. Clearly $\mathcal{I}$ is connectable if and only if there exists $S\subseteq \sm(r)$ such that $f(r,S)=1$.

We compute all $f(v,S)$ and $g(v,c)$ by dynamic programming in a bottom-up manner. Initially all values to set to 0. The values for leaf nodes are easy to obtain. Assume we want to compute $f(v,S)$, given that the values of $f(v',S')$ and $g(v',c)$ are all known if $v'$ is a child of $v$. Then $f(v,S)=1$ if and only if for every child $v'$ of $v$, there exists $c\in S$ such that $g(v',c)=1$ (in which case $v$ and $v'$ are connected through channel $c$). If $f(v,S)$ turns out to be 1, we set $g(v,c)$ to 1 for all $c\in S$. It is easy to see that $g(v,c)$ will be correctly computed after the values of $f(v,S)$ are obtained for all possible $S$. After all values have been computed, we check whether $f(r,S)=1$ for some $S\subseteq \sm(r)$.

Recall that $n=|V|$, $k=|C|$, and denote $t=\max_{v\in V}|\sm(v)|$. There are at most $n(2^{t}+k)$ terms to be computed, each of which takes time ${\rm poly}(n,k)$ by our previous analysis. The final checking step takes $2^{t}{\rm poly}(n,k)$ time. Hence the total running time is $2^{t}{\rm poly}(n,k)=2^{t}(kn)^{O(1)}$, which is at most $2^{O(k)}n^{O(1)}$ since $t\leq k$.
Finally note that it is easy to modify the algorithm so that, given a connectable network it will return a \specassign that makes it connected. \qed
\end{proof}

%The following corollary is straightforward from Theorem~\ref{thm:alg_tree}.

\begin{corollary}\label{cor:fpt_tree}
The \fullprob problem with trees as \pgraphs is fixed parameter tractable when parameterized by the number of
channels.
\end{corollary}

\subsection{Bounded Treewidth Graphs}
In this part we deal with another class of potential graphs, namely the class of graphs with bounded treewidth.
%In addition to the number of SUs and channels, the running time of our algorithm also depends on another natural parameter, namely the treewidth of the \pgraph.
Our main result is the following theorem, which generalizes Theorem~\ref{thm:alg_tree} as a tree has treewidth one.

\begin{theorem}\label{thm:alg_treewidth}
%For any fixed integer $\tw \geq 1$, the \fullprob problem can be solved in $2^{O(\tw \cdot k)}n^{O(1)}$ time on the class of potential graphs that have treewidth at most $\tw$.
There is an algorithm that, given a \network whose \pgraph has bounded treewidth, checks whether it is
connectable in $2^{O(k)}n^{O(1)}$ time.
\end{theorem}

\begin{proof}
Suppose we are given a \network $\mathcal{I}$ with \pgraph $G=(V,E)$, which has treewidth $\tw=O(1)$.
Let $(T=(I,F),\{X_i~|~i\in I\})$ be a nice tree decomposition of $G$ of width $\tw$ (see Section~\ref{subsec:treedecomp} for the related notions). Recall that $T$ is a rooted binary tree with $O(|V|)$ nodes and can be found in polynomial time. Let $r$ be the root of $T$. For every non-leaf node $i$ of $T$, let $i_L$ and $i_R$ be the two children of $i$. (We can always add dummy leaf-nodes to make every non-leaf node have exactly two children, which at most doubles the size of $T$.)

For each $i\in I$, define
$$Y_i:=\{v\in X_j~|~j=i \textrm{~or~} j \textrm{~is a descendent of~}i\},$$
and let $\mathcal{I}_i$ be a new instance of the problem that is almost identical to $\mathcal{I}$ except that we replace the \pgraph with $G[Y_i]$, i.e., the subgraph of $G$ induced on the vertex set $Y_i \subseteq V$.

For each $i\in I$, suppose $X_i=\{v_1,v_2,\ldots,v_{t}\}$ where $t=|X_i|$ and $v_j\in V$ for all $1\leq j\leq t$.
For each tuple $(S_1,S_2,\ldots,S_t)$ such that $S_j \subseteq \sm(v_j)$ for all $1\leq j\leq t$, we use a Boolean variable
$\mathcal{B}_i(S_1, S_2,\ldots,S_t)$ to indicate whether there exists a spectrum assignment $\mathcal{SA}_i$ that makes $\mathcal{I}_i$ connected such that $\mathcal{SA}_i(v_j)=S_j$ for all $1\leq j\leq t$.
Notice that for each $i$, the number of such variables is at most $(2^k)^{|X_i|}\leq 2^{k\cdot \tw}$, and we can list  them in $2^{O(k\cdot \tw)}$ time.
Initially all variables are set to FALSE.
Assume $X_r=\{w_1,w_2,\ldots,w_{|X_r|}\}$ (recall that $r$ is the root of $T$). Then, clearly, deciding whether $\mathcal{I}$ is connectable is equivalent to checking whether there exists $(S_1,S_2,\ldots,S_{|X_r|})$, where $S_j\subseteq \sm(w_j)$ for all $1\leq j\leq |X_r|$, such that $\mathcal{B}_r(S_1,S_2,\ldots,S_{|X_r|})$ is TRUE.

We will compute the values of all possible $\mathcal{B}_i(S_1,S_2,\ldots,S_t)$ by dynamic programming.
For each leaf node $l$, we can compute the values of all the variables related to $\mathcal{I}_l$ in time $2^{O(k\cdot \tw)}n^{O(1)}$ by the brute-force approach.

Now suppose we want to decide the value of $\mathcal{B}_i(S_1,S_2,\ldots,S_{|X_i|})$ for some non-leaf node $i$, provided that the variables related to any children of $i$ have all been correctly computed. Recall that $i_L$ and $i_R$ are the two children of $i$. We define:
 \begin{itemize}
\item $NEW=X_i \setminus (Y_{i_L}\cup Y_{i_R})$;
\item $OLD=X_i \setminus NEW = X_i \cap (Y_{i_L}\cup Y_{i_R})$;
\item $Z_L=Y_{i_L} \setminus X_i$, and $Z_R = Y_{i_R} \setminus X_i$.
\end{itemize}
It is clear that $Y_i = NEW \cup OLD \cup Z_L \cup Z_R$.
By using the properties of a tree decomposition, we have the following fact:

\begin{lemma}\label{lem:disjoint}
$NEW$, $Z_L$, and $Z_R$ are three pairwise disjoint subsets of $V$, and there is no edge of $G$ whose endpoints lie in different subsets.
\end{lemma}
\begin{proof}
 Since $NEW \subseteq X_i$ and $Z_L=Y_{i_L}\setminus X_i$, we have $NEW \cap Z_L=\emptyset$, and similarly $NEW \cap Z_R = \emptyset$. Assume that $Z_L \cap Z_R \neq \emptyset$, and let $v\in Z_L \cap Z_R$. Since $Z_L \subseteq Y_{i_L}$ and $Z_R \subseteq Y_{i_R}$, we have $v\in Y_{i_L} \cap Y_{i_R}$. By the definition of a tree decomposition, $v\in X_i$, so $v \in X_i \cap Z_L = X_i \cap (Y_{i_L}\setminus X_i)=\emptyset$, a contradiction. Therefore $Z_L \cap Z_R=\emptyset$. This proves the pairwise disjointness of the three sets.

 Now assume that there exists an edge $e=(u,v)\in E$ such that $u\in Z_L$ and $v\in Z_R$. Then, by the definition of a tree decomposition, there exists $p \in I$ such that $\{u,v\}\subseteq X_p$. We know that $p\neq i$. So there are three possibilities: $p$ lines in the subtree rooted at $X_{i_L}$, or in the subtree rooted at $X_{i_R}$, or it is not in the subtree rooted at $X_i$. It is easy to verify that, in each of the three cases, we can find a path that connects two tree nodes both containing $u$ (or $v$) and goes through $i$, which implies $u \in X_i$ or $v\in X_i$ by the property of a tree decomposition. This contradicts our previous result. Thus there is no edge with one endpoint in $Z_L$ and another in $Z_R$. Similarly, we can prove that there exists no edge with one endpoint in $NEW$ and another in $Z_L$ or $Z_R$. This completes the proof of the lemma.
 \qed
\end{proof}

We now continue the proof of Theorem~\ref{thm:alg_treewidth}. Recall that we want to decide $\mathcal{B}_i(S_1,S_2,\ldots,S_{|X_i|})$, i.e., whether $\mathcal{I}_i$, the network with $G[Y_i]$ as the \pgraph, is connectable under some spectrum assignment $\mathcal{SA}$ such that $\mathcal{SA}(v_j)=S_j$ for all $1\leq j\leq |X_i|$ (we assume that $X_i=\{v_1,v_2,\ldots,v_{|X_i|}\}$). Note that $Y_i = NEW \cup OLD \cup Z_L \cup Z_R$. Due to Lemma~\ref{lem:disjoint}, the three subsets $NEW$, $Z_L$ and $Z_R$ can only be connected through $OLD$ (or, we can think $OLD$ as an ``intermediate'' set). Therefore, for any spectrum assignment $\mathcal{SA}$ such that $\mathcal{SA}(v_j)=S_j$ for all $j$, $\mathcal{I}_i$ is connected under $\mathcal{SA}$ if and only if the following three things simultaneously hold:

\begin{itemize}
\item $G[X_i]$ is connected under $\mathcal{SA}$;
\item $\mathcal{B}_{i_L}(S'_1,\ldots,S'_{|X_{i_L}|})$ is TRUE for some $(S'_1,\ldots,S'_{|X_{i_L}|})$ that accords with $(S_1,\ldots,S_{|X_i|})$, i.e., the two vectors coincide on any component corresponding to a vertex in $X_i \cap X_{i_L}$;
\item $\mathcal{B}_{i_R}(S'_1,\ldots,S'_{|X_{i_R}|})$ is TRUE for some $(S'_1,\ldots,S'_{|X_{i_R}|})$ that accords with $(S_1,\ldots,S_{|X_i|})$, i.e., the two vectors coincide on any component corresponding to a vertex in $X_i \cap X_{i_R}$.
\end{itemize}

The first condition above can be checked in polynomial time, and the last two conditions can be verified in $2^{O(k\cdot \tw)}n^{O(1)}$ time. Thus the time spent on determining $\mathcal{B}_i(S_1,\ldots,S_{|X_i|})$ is $2^{O(k\cdot \tw)}n^{O(1)}$. After all such terms have been computed, we can get the correct answer by checking whether there exists $(S_1,\ldots,S_{|X_r|})$ such that $\mathcal{B}_r(S_1,\ldots,S_{|X_r|})$ is TRUE, which costs another $2^{O(k\cdot \tw)}n^{O(1)}$ time. Since there are at most $O(|V|)=O(n)$ nodes in $T$, the total running time of the algorithm is $2^{O(k\cdot \tw)}n^{O(1)}=2^{O(k)}n^{O(1)}$ as $\tw=O(1)$. The proof is complete.
\qed
\end{proof}

\begin{corollary}\label{cor:fpt_tree}
The \fullprob problem on bounded treewidth graphs is fixed parameter tractable when parameterized by the number of
channels.
\end{corollary}

\section{Conclusion and Future Work}\label{sec:con}
In this paper, we initiate a systematic study on the algorithmic complexity of
connectivity problem in cognitive radio networks through spectrum
assignment. The hardness of the problem in the general case and
several special cases are addressed, and exact algorithms are
also derived to check whether the network is connectable.
%Our work gives a better understanding of the complexity of the problem.

In some applications, when the given \network is not connectable, we may want to connect the largest subset of the secondary users. This optimization problem is NP-hard, since the decision version is already NP-complete on very restricted instances. Thus it is interesting to design polynomial time approximation algorithms for this optimization problem.

In some other scenarios, we may wish to connect all the secondary users but keep the \budget as low as possible. That is, we want to find the smallest $\beta$ such that there exists a \specassign\ connecting the graph in which each SU opens at most $\beta$ channels. It is easy to see that this problem generalizes the minimum-degree spanning tree problem \cite{book_npc}, which asks to find a spanning tree of a given graph in which the maximum vertex degree is minimized. The latter problem is NP-hard, but there is a polynomial time algorithm that finds a spanning tree of degree at most one more than the optimum \cite{DBLP:journals/jal/FurerR94}. It would be interesting to see whether this algorithm can be generalized to the min-budget version of our connectivity problem, or whether we can at least obtain constant factor approximations.

Another meaningful extension of this work is to design distributed algorithms to achieve network connectivity. Moreover, due to interference in wireless communications, the connected nodes using the same channel may not be able to communicate simultaneously. Therefore, it is also interesting to investigate distributed algorithms with channel assignment and link scheduling jointly considered to achieve some network objective such as connectivity and capacity maximization, especially under the realistic interference models.

%explore how to schedule the links such that the network throughput is optimized
%under realistic interference models. Finally, as many real networks are distributed rather than centralized, it is desirable to design efficient distributed algorithms to achieve network connectivity.

\section*{Acknowledgements}
The authors would like to give thanks to Dr. Thomas Moscibroda at Microsoft
Research Asia for his introduction of the original problem. This work was
supported in part by the National Basic Research Program of China Grant
2011CBA00300, 2011CBA00302, the National Natural Science Foundation of
China Grant 61073174, 61103186, 61202360, 61033001, and 61061130540.

\bibliographystyle{abbrv}
\bibliography{connectivity}

\begin{thebibliography}{10}

\bibitem{sigcomm09Bahl}
P.~Bahl, R.~Chandra, T.~Moscibroda, R.~Murty, and M.~Welsh.
\newblock White space networking with {W}i-{F}i like connectivity.
\newblock In {\em ACM SIGCOMM}, 2009.

\bibitem{bod_sjc96}
H.~L. Bodlaender.
\newblock A linear-time algorithm for finding tree-decompositions of small
  treewidth.
\newblock {\em SIAM J. Comput.}, 25(6):1305--1317, 1996.

\bibitem{cayley}
A.~Cayley.
\newblock A theorem on trees.
\newblock {\em Quarterly Journal on Pure and Applied Mathematics}, 23:376--378,
  1889.

\bibitem{discDolev}
S.~Dolev, S.~Gilbert, R.~Guerraoui, and C.~C. Newport.
\newblock Gossiping in a multi-channel radio network.
\newblock In {\em DISC}, 2007.

\bibitem{book_para}
R.~Downey and M.~Fellows.
\newblock {\em Parameterized Complexity}.
\newblock Springer, 1998.

\bibitem{DBLP:journals/jal/FurerR94}
M.~F{\"u}rer and B.~Raghavachari.
\newblock Approximating the minimum-degree steiner tree to within one of
  optimal.
\newblock {\em J. Algorithms}, 17(3):409--423, 1994.

\bibitem{spanntree}
H.~Gabow and E.~Myers.
\newblock Finding all spanning trees of directed and undirected graphs.
\newblock {\em SIAM J. Comput.}, 7(3):280--287, 1978.

\bibitem{book_npc}
M.~Garey and D.~Johnson.
\newblock {\em Computers and Intractability: A Guide to the Theory of
  NP-Completeness}.
\newblock W. H. Freeman, 1979.

\bibitem{bmatching}
J.~Hopcroft and R.~Karp.
\newblock An $n^{5/2}$ algorithm for maximum matchings in bipartite graphs.
\newblock {\em SIAM J. Comput.}, 2(4):225--231, 1973.

\bibitem{kloks94}
T.~Kloks.
\newblock {\em Treewidth: Computations and Approximations}, volume 842 of {\em
  LNCS}.
\newblock Springer, 1994.

\bibitem{wireless08Kuhn}
F.~Kuhn, R.~Wattenhofer, and A.~Zollinger.
\newblock Ad hoc networks beyond unit disk graphs.
\newblock {\em Wireless Netwoks}, 14(5):715--729, 2008.

\bibitem{infocom12Li}
X.-Y. Li, P.~Yang, Y.~Yan, L.~You, S.~Tang, and Q.~Huang.
\newblock Almost optimal accessing of nonstochastic channels in cognitive radio
  networks.
\newblock In {\em IEEE INFOCOM}, 2012.

\bibitem{algosensor12}
H.~Liang, T.~Lou, H.~Tan, Y.~Wang, and D.~Yu.
\newblock Complexity of connectivity in cognitive radio networks through
  spectrum assignment.
\newblock In {\em ALGOSENSORS'12}, volume 7718 of {\em LNCS}, pages 108--119,
  2013.

\bibitem{Infocom12Lu}
D.~Lu, X.~Huang, P.~Li, and J.~Fan.
\newblock Connectivity of large-scale cognitive radio ad hoc networks.
\newblock In {\em IEEE INFOCOM}, 2012.

\bibitem{CoRoNetRen}
W.~Ren, Q.~Zhao, and A.~Swami.
\newblock Connectivity of cognitive radio networks: Proximity vs. opportunity.
\newblock In {\em ACM CoRoNet}, 2009.

\bibitem{JsacRen}
W.~Ren, Q.~Zhao, and A.~Swami.
\newblock Power control in cognitive radio networks: How to cross a multi-lane
  highway.
\newblock {\em IEEE JSAC}, 27(7):1283--1296, 2009.

\bibitem{rs_ja86}
N.~Robertson and P.~D. Seymour.
\newblock Graph minors. ii. algorithmic aspects of treewidth.
\newblock {\em J. Algorithms}, 7:309--322, 1986.

\bibitem{cayley_arxiv}
A.~Shukla.
\newblock A short proof of {C}ayley's tree formula.
\newblock {\em arXiv:0908.2324v2}.

\bibitem{suveryWang}
B.~Wang and K.~J.~R. Liu.
\newblock Advances in cognitive radio networks: {A} survey.
\newblock {\em IEEE J. Selected Topics in Signal Processing}, 5(1):5--23, 2011.

\bibitem{mobihoc12Huang}
C.~Xu and J.~Huang.
\newblock Spatial spectrum access game: nash equilibria and distributed
  learning.
\newblock In {\em ACM MobiHoc}, 2012.

\bibitem{jsac12Du}
J.~Yu, H.~Roh, W.~Lee, S.~Pack, and D.-Z. Du.
\newblock Topology control in cooperative wireless ad-hoc networks.
\newblock {\em IEEE JSAC}, 30(9):1771--1779, 2012.

\bibitem{mobihoc07Yuan}
Y.~Yuan, P.~Bahl, R.~Chandra, T.~Moscibroda, and Y.~Wu.
\newblock Allocating dynamic time-spectrum blocks in cognitive radio networks.
\newblock In {\em ACM MobiCom}, 2007.

\bibitem{auctionZhou}
X.~Zhou, S.~Gandhi, S.~Suri, and H.~Zheng.
\newblock ebay in the sky: Strategyproof wireless spectrum auctions.
\newblock In {\em ACM MobiCom}, 2008.

\end{thebibliography}

\end{document}